\documentclass[twoside, conference]{IEEEtran}

\usepackage{epsfig}
\usepackage{graphicx}
\usepackage{amsmath}
\usepackage{amssymb}
\usepackage{float}
\usepackage{url}
\newtheorem{definition}{Definition}

\newtheorem{proposition}{Proposition}
\newtheorem{corollary}{Corollary}
\newtheorem{theorem}{Theorem}

\newtheorem{remark}{Remark}

\newcommand{\naturals}{\ensuremath{\mathbb{N}}}
\newcommand{\reals}{\ensuremath{\mathbb{R}}}
\newcommand{\pr}{\ensuremath{\mathbb{P}}}
\newcommand{\expectation}{\ensuremath{\mathbb{E}}}

\begin{document}
\markboth{SUBMITTED TO THE Eighth International Symposium on
Wireless Communication Systems, Aachen, Germany, 6--9 November,
2011} {I. SASON: On Concentration of the Crest Factor for OFDM
Signals}

\title{On the Concentration of the Crest Factor for OFDM Signals}

\author{Igal~Sason\\Department of Electrical Engineering\\
Technion - Israel Institute of Technology, Haifa 32000, Israel\\
E-mail: sason@ee.technion.ac.il}

\maketitle

\begin{abstract}
This paper applies several concentration inequalities to prove
concentration results for the crest factor of OFDM signals. The
considered approaches are, to the best of our knowledge, new in
the context of establishing concentration for OFDM signals.
\end{abstract}


\section{Introduction}
\label{section: introduction}
Orthogonal-frequency-division-multiplexing (OFDM) is a modulation
that converts a high-rate data stream into a number of low-rate
steams that are transmitted over parallel narrow-band channels.
OFDM is widely used in several international standards for digital
audio and video broadcasting, and for wireless local area
networks. For a textbook providing a survey on OFDM, see e.g.
\cite[Chapter~19]{Molisch_book}.

One of the problems of OFDM is that the peak amplitude of the
signal can be significantly higher than the average amplitude.
This issue makes the transmission of OFDM signals sensitive to
non-linear devices in the communication path such as digital to
analog converters, mixers and high-power amplifiers. As a result
of this drawback, it increases the symbol error rate and it also
reduces the power efficiency of OFDM signals as compared to
single-carrier systems. Commonly, the impact of nonlinearities is
described by the distribution of the crest-factor (CF) of the
transmitted signal \cite{LitsynW06}, but its calculation involves
time-consuming simulations even for a small number of
sub-carriers. The expected value of the CF for OFDM signals is
known to scale like the logarithm of the number of sub-carriers of
the OFDM signal (see \cite{LitsynW06}, \cite[Section~4]{SalemZ}
and \cite{WunderB_IT}).

In this paper, we consider two of the main approaches for proving
concentration inequalities, and apply them to derive concentration
results for the crest factor of OFDM signals. The first approach
is based on martingales, and the other approach is Talagrand's
method for proving concentration inequalities in product spaces.
It is noted that some of these concentration inequalities can be
derived using ideas from information theory (see, e.g.,
\cite{McDiarmid_tutorial} and references therein).

Considering the martingale approach for proving concentration
results, the Azuma-Hoeffding inequality is by now a well-known
methodology that has been often used to prove concentration of
measures. It is due to Hoeffding \cite{Hoeffding} who proved this
inequality for a sum of independent and bounded random variables,
and Azuma \cite{Azuma} who later extended it to bounded-difference
martingales. In the context of communication and information
theoretic aspects, Azuma's inequality was used during the last
decade in the coding literature for establishing concentration
results for codes defined on graphs and iterative decoding
algorithms (see, \cite{RiU_book} and references therein). Some
other martingale-based concentration inequalities were also
recently applied to the performance evaluation of random coding
over non-linear communication channels \cite{Volterra_IT_March11}.
McDiarmid's inequality is an improved version of Azuma's
inequality in the special case where one considers the
concentration of a function $f: \reals^n \rightarrow \reals$ of
$n$ independent RVs when the variation of $f(x_1, \ldots, x_n)$
w.r.t. each of its coordinates is bounded (and when all the other
$n-1$ components are kept fixed). Under this setting, it gives an
improvement by a factor of~4 in the exponent. This inequality is
applied in this paper in order to prove the concentration of the
crest factor of OFDM signals around the expected value.

A second approach for proving concentration inequalities in
product spaces was developed by Talagrand in his seminal paper
\cite{Talagrand95}. It forms in general a powerful probabilistic
tool for establishing concentration results for coordinate-wise
Liphschitz functions of independent random variables (see also,
e.g., \cite[Section~2.4.2]{Dembo_Zeitouni} and
\cite[Section~4]{McDiarmid_tutorial}). This approach was used in
\cite{Korada_Macris_IT10} to prove concentration inequalities, in
the large system limit, for a code division multiple access (CDMA)
system. Talagrand's inequality is used in this paper to prove a
concentration result (near the median) of the crest factor of OFDM
signals, and it also enables to obtain an upper bound on the
distance between the median and the expected value.

A stronger concentration inequality for the crest factor of OFDM
signals was introduced in \cite[Theorem~3]{LitsynW06} under some
assumptions on the probability distribution of the considered
problem (the reader is referred to the two conditions in
\cite[Theorem~3]{LitsynW06}, followed by
\cite[Corollary~5]{LitsynW06}). These requirements are not needed
in the following analysis, and the derivation of the concentration
inequalities here is rather simple and it provides some further
insight to this issue.

\section{Some Concentration Inequalities}
In the following, we present briefly essential background on
concentration inequalities that is required for the analysis in
this paper. In the next section, we will apply these probabilistic
tools for obtaining concentration inequalities for the crest
factor of OFDM signals.

\subsection{Azuma's Inequality} \label{subsection:
Azuma's inequality} Azuma's inequality\footnote{Azuma's inequality
is also known as the Azuma-Hoeffding inequality. Since this
inequality is referred several times in this paper, it will be
named from this point as Azuma's inequality for the sake of
brevity.} forms a useful concentration inequality for
bounded-difference martingales \cite{Azuma}. In the following,
this inequality is introduced.

\begin{theorem}{\bf[Azuma's inequality]}
Let $\{X_k, \mathcal{F}_k\}_{k=0}^{\infty}$ be a
discrete-parameter real-valued martingale sequence such that for
every $k \in \naturals$, the condition $ |X_k - X_{k-1}| \leq d_k$
holds a.s. for some non-negative constants
$\{d_k\}_{k=1}^{\infty}$. Then
\begin{equation}
\pr( | X_n - X_0 | \geq r) \leq 2 \exp\left(-\frac{r^2}{2
\sum_{k=1}^n d_k^2}\right) \, \quad \forall \, r \geq 0.
\label{eq: Azuma's concentration inequality - general case}
\end{equation}
\label{theorem: Azuma's concentration inequality}
\end{theorem}

The concentration inequality stated in Theorem~\ref{theorem:
Azuma's concentration inequality} was proved in \cite{Hoeffding}
for independent bounded random variables, followed by a discussion
on sums of dependent random variables; this inequality was later
derived in \cite{Azuma} for bounded-difference martingales. The
reader is referred, e.g., to \cite{McDiarmid_tutorial} for a
proof.

\subsection{A Refined Version of Azuma's Inequality}
The following refined version of Azuma's inequality was introduced
in \cite{Sason_submitted_paper} (which includes some other
approaches for refining Azuma's inequality).
\begin{theorem}
Let $\{X_k, \mathcal{F}_k\}_{k=0}^{\infty}$ be a
discrete-parameter real-valued martingale. Assume that, for some
constants $d, \sigma > 0$, the following two requirements are
satisfied a.s.
\begin{eqnarray*}
&& | X_k - X_{k-1} | \leq d, \\
&& \text{Var} (X_k | \mathcal{F}_{k-1}) = \expectation \bigl[(X_k
- X_{k-1})^2 \, | \, \mathcal{F}_{k-1} \bigr] \leq \sigma^2
\end{eqnarray*}
for every $k \in \{1, \ldots, n\}$. Then, for every $\alpha \geq
0$,
\begin{equation}
\hspace*{-0.2cm} \pr(|X_n-X_0| \geq \alpha n) \leq 2 \exp\left(-n
\, D\biggl(\frac{\delta+\gamma}{1+\gamma} \Big|\Big|
\frac{\gamma}{1+\gamma}\biggr) \right) \label{eq: first refined
concentration inequality}
\end{equation}
where
\begin{equation}
\gamma \triangleq \frac{\sigma^2}{d^2}, \quad \delta \triangleq
\frac{\alpha}{d}  \label{eq: notation}
\end{equation}
and
\begin{equation}
D(p || q) \triangleq p \ln\Bigl(\frac{p}{q}\Bigr) + (1-p)
\ln\Bigl(\frac{1-p}{1-q}\Bigr), \quad \forall \, p, q \in [0,1]
\label{eq: divergence}
\end{equation}
is the divergence (a.k.a. relative entropy or Kullback-Leibler
distance) between the two probability distributions $(p,1-p)$ and
$(q,1-q)$. If $\delta>1$, then the probability on the left-hand
side of \eqref{eq: first refined concentration inequality} is
equal to zero. \label{theorem: first refined concentration
inequality}
\end{theorem}
\begin{proof}
The idea of the proof of Theorem~\ref{theorem: first refined
concentration inequality} is essentially similar to the proof of
\cite[Corollary~2.4.7]{Dembo_Zeitouni}. The full proof is provided
in \cite[Section~III]{Sason_submitted_paper}.
\end{proof}

\begin{proposition}
Let $\{X_k, \mathcal{F}_k\}_{k=0}^{\infty}$ be a
discrete-parameter real-valued martingale. Then, for every $\alpha
\geq 0$,
\begin{equation}
\pr(|X_n-X_0| \geq \alpha \sqrt{n}) \leq 2
\exp\Bigl(-\frac{\delta^2}{2\gamma}\Bigr) \Bigl(1+
O\bigl(n^{-\frac{1}{2}}\bigr)\Bigr) \label{eq: concentration1}
\end{equation}
where $\gamma$ and $\delta$ are introduced in \eqref{eq:
notation}. \label{proposition: a similar scaling of the
concentration inequalities}
\end{proposition}
\begin{proof}
This inequality follows from Theorem~\ref{theorem: first refined
concentration inequality} (see
\cite[Appendix~H]{Sason_submitted_paper}).
\end{proof}

\subsection{McDiarmid's Inequality}
In the following, we state McDiarmid's inequality (see
\cite[Theorem~3.1]{McDiarmid_tutorial}).

\begin{theorem}
Let ${\bf{X}} = (X_1, \ldots, X_n)$ be a family of independent
random variables with $X_k$ taking values in a set $A_k$ for each
$k$. Suppose that a real-valued function $f$, defined on $\prod_k
A_k$, satisfies
\begin{equation*}
| f({\bf{x}}) - f({\bf{x'}}) | \leq c_k
\end{equation*}
whenever the vectors ${\bf{x}}$ and ${\bf{x'}}$ differ only in the
$k$-th coordinate. Let $\mu \triangleq \expectation[f(X)]$ be the
expected value of $f(X)$. Then, for every $\alpha \geq 0$,
\begin{equation*}
\pr(| f(X) - \mu | \geq \alpha) \leq 2
\exp\left(-\frac{2\alpha^2}{\sum_k c_k^2} \right).
\end{equation*}
\label{theorem: McDiarmid's Inequality}
\end{theorem}
This inequality is proved with the aid of martingales. It has some
nice applications which were exemplified in the context of
algorithmic discrete mathematics (see
\cite[Section~3]{McDiarmid_tutorial}).

\subsection{Talagrand's inequality}
Talagrand's inequality is an approach used for establishing
concentration results on product spaces, and this technique was
introduced in Talagrand's landmark paper \cite{Talagrand95}.

We provide in the following two definitions that will be required
for the introduction of a special form of Talagrand's
inequalities.

\begin{definition}[Hamming distance]
Let ${\bf{x}}, {\bf{y}}$ be two $n$-length vectors. The Hamming
distance between ${\bf{x}}$ and ${\bf{y}}$ is the number of
coordinates where ${\bf{x}}$ and ${\bf{y}}$ disagree, i.e.,
\begin{equation*}
d_{\text{H}}({\bf{x}},{\bf{y}}) \triangleq \sum_{i=1}^n I_{\{x_i
\neq y_i\}}
\end{equation*}
where $I$ stands for the indicator function.
\end{definition}

The following suggests a generalization and normalization of the
previous distance metric.
\begin{definition}
Let $a = (a_1, \ldots, a_n) \in \reals_{+}^n$ (i.e., $a$ is a
non-negative vector) satisfy $||a||^2 = \sum_{i=1}^n (a_i)^2 = 1$.
Then, define
\begin{equation*}
d_a({\bf{x}},{\bf{y}}) \triangleq \sum_{i=1}^n a_i I_{\{x_i \neq
y_i\}}.
\end{equation*}
Hence, $d_{\text{H}}({\bf{x}},{\bf{y}}) = \sqrt{n} \,
d_a({\bf{x}},{\bf{y}})$ for $a = \bigl(\frac{1}{\sqrt{n}}, \ldots,
\frac{1}{\sqrt{n}}\bigr)$.
\end{definition}

The following is a special form of Talagrand's inequalities
(\cite[Chapter~4]{McDiarmid_tutorial}, \cite{Talagrand95},
\cite{Talagrand96}).
\begin{theorem}[Talagrand's inequality]
Let the random vector ${\bf{X}} = (X_1, \ldots, X_n)$ be a vector of
independent random variables with $X_k$ taking values in a set
$A_k$, and let $A \triangleq \prod_{k=1}^n A_k$. Let $f: A
\rightarrow \reals$ satisfy the condition that, for every
${\bf{x}} \in A$, there exists a non-negative, normalized
$n$-length vector $a = a(x)$ such that
\begin{equation}
f({\bf{x}}) \leq f({\bf{y}}) + \sigma d_a({\bf{x}}, {\bf{y}}),
\quad \forall \, {\bf{y}} \in A  \label{eq: Talagrand's condition}
\end{equation}
for some fixed value $\sigma>0$. Then, for every $\alpha \geq 0$,
\begin{equation}
\pr(|f(X) - m| \geq \alpha) \leq 4 \exp\left(-\frac{\alpha^2}{4
\sigma^2}\right)  \label{eq: Talagrand's inequality}
\end{equation}
where $m$ is the median of $f(X)$ (i.e., $\pr(f(X) \leq m) \geq
\frac{1}{2}$ and $\pr(f(X) \geq m) \geq \frac{1}{2}$). The same
conclusion in \eqref{eq: Talagrand's inequality} holds if the
condition in \eqref{eq: Talagrand's condition} is replaced by
\begin{equation}
f({\bf{y}}) \leq f({\bf{x}}) + \sigma d_a({\bf{x}}, {\bf{y}}),
\quad \forall \, {\bf{y}} \in A.  \label{eq: Talagrand's 2nd
condition}
\end{equation}
\label{theorem: Talagrand's Inequality}
\end{theorem}

\begin{remark} In the special case where the condition for the
function $f$ in Theorem~\ref{theorem: Talagrand's Inequality}
(Talagrand's inequality) is satisfied with the additional property
that the vector $a$ on the right-hand side of \eqref{eq:
Talagrand's condition} is {\em independent} of $x$ (i.e., the
value of this vector is fixed), then the concentration inequality
in \eqref{eq: Talagrand's inequality} follows from McDiarmid's
inequality. To verify this observation, the reader is referred to
\cite[Theorem~3.6]{McDiarmid_tutorial} followed by the discussion
in \cite[p.~211]{McDiarmid_tutorial} (leading to \cite[Eqs.~(3.12)
and (3.13)]{McDiarmid_tutorial}. \label{remark: on a relation
between McDiardmid's inequality and a certain form of Talagrand's
inequality}
\end{remark}

\section{Application: Concentration of the Crest-Factor for OFDM Signals}
\label{section: OFDM}

\subsection{Background}
\label{subsection: background on OFDM} Given an $n$-length
codeword $\{X_i\}_{i=0}^{n-1}$, a single OFDM baseband symbol is
described by
\begin{equation}
s(t; X_0, \ldots, X_{n-1}) = \frac{1}{\sqrt{n}} \sum_{i=0}^{n-1}
X_i \exp\Bigl(\frac{j \, 2\pi i t}{T}\Bigr), \quad 0 \leq t \leq
T.  \label{eq: OFDM signal}
\end{equation}
Lets assume that $X_0, \ldots, X_{n-1}$ are i.i.d. complex RVs
with $|X_i|=1$. Since the sub-carriers are orthonormal over
$[0,T]$, then a.s. the power of the signal $s$ over this interval
is~1. The CF of the signal $s$, composed of $n$ sub-carriers, is
defined as
\begin{equation}
\text{CF}_n(s) \triangleq \max_{0 \leq t \leq T} | s(t) |.
\label{eq: CF}
\end{equation}
From \cite[Section~4]{SalemZ} and \cite{WunderB_IT}, it follows
that the CF scales with high probability like $\sqrt{\log(n)}$ for
large $n$. In \cite[Theorem~3 and Corollary~5]{LitsynW06}, a
concentration inequality was derived for the CF of OFDM signals.
It states that for every $c \geq 2.5$

\vspace*{-0.4cm} \small \begin{equation*} \pr \biggl(\Bigl|
\text{CF}_n(s) - \sqrt{\log(n)} \Bigr| < \frac{c \log \log
(n)}{\sqrt{\log(n)}} \biggr) = 1 - O\Biggl( \frac{1}{\bigl(
\log(n) \bigr)^4} \Biggr).
\end{equation*}
\normalsize
\begin{remark}
The analysis used to derive this rather strong concentration
inequality (see \cite[Appendix~C]{LitsynW06}) requires some
assumptions on the distribution of the $X_i$'s (see the two
conditions in \cite[Theorem~3]{LitsynW06} followed by
\cite[Corollary~5]{LitsynW06}). These requirements are not needed
in the following analysis, and the derivation of the two
concentration inequalities in this paper is simple, though weaker
concentration results are obtained.
\end{remark}

\vspace*{0.2cm} In the following, Azuma's inequality and a refined
version of this inequality are considered under the assumption
that $\{X_j\}_{j=0}^{n-1}$ are independent complex-valued random
variables with magnitude~1, attaining the $M$ points of an $M$-ary
PSK constellation with equal probability.

\subsection{Establishing Concentration of the Crest-Factor via Azuma's Inequality and a Refined Version}

\subsubsection{Proving Concentration via Azuma's
Inequality} In the following, Azuma's inequality is used to derive
a concentration result. Let us define
\begin{equation}
Y_i = \expectation[ \, \text{CF}_n(s) \, | \, X_0, \ldots,
X_{i-1}], \quad i =0, \ldots, n  \label{eq: martingale sequence
for OFDM}
\end{equation}
Based on a standard construction of Doob's martingales, $\{Y_i,
\mathcal{F}_i\}_{i=0}^n$ is a martingale where $\mathcal{F}_i$ is
the $\sigma$-algebra that is generated by the first $i$ symbols
$(X_0, \ldots, X_{i-1})$ in \eqref{eq: OFDM signal}. Hence,
$\mathcal{F}_0 \subseteq \mathcal{F}_1 \subseteq \ldots \subseteq
\mathcal{F}_n$ is a filtration. This martingale has also bounded
jumps, and
$$|Y_i - Y_{i-1}| \leq \frac{2}{\sqrt{n}}$$
for $i \in \{1, \ldots, n\}$ since revealing the additional $i$-th
coordinate $X_i$ affects the CF, as is defined in \eqref{eq: CF},
by at most $\frac{2}{\sqrt{n}}$ (see the first part of
Appendix~\ref{appendix: OFDM}). It therefore follows from Azuma's
inequality that, for every $\alpha > 0$,
\begin{equation}
\hspace*{-0.2cm} \pr( | \text{CF}_n(s) - \expectation [
\text{CF}_n(s)] | \geq \alpha) \leq 2
\exp\left(-\frac{\alpha^2}{8}\right) \label{eq: Azuma's inequality
for OFDM}
\end{equation}
which demonstrates the concentration of this measure around its
expected value.

\vspace*{0.2cm}
\subsubsection{Proof of Concentration via
Proposition~\ref{proposition: a similar scaling of the
concentration inequalities}} In the following, we rely on
Proposition~\ref{proposition: a similar scaling of the
concentration inequalities} to derive an improved concentration
result. For the martingale sequence $\{Y_i\}_{i=0}^n$ in
\eqref{eq: martingale sequence for OFDM}, Appendix~\ref{appendix:
OFDM} gives that a.s.
\begin{equation}
|Y_i - Y_{i-1}| \leq \frac{2}{\sqrt{n}} \, , \quad
\expectation\bigl[(Y_i-Y_{i-1})^2  | \mathcal{F}_{i-1}\bigr] \leq
\frac{2}{n}
\label{eq: properties for OFDM signals}
\end{equation}
for every $i \in \{1, \ldots, n\}$. Note that the conditioning on
the $\sigma$-algebra $\mathcal{F}_{i-1}$ is equivalent to the
conditioning on the symbols $X_0, \ldots, X_{i-2}$, and there is
no conditioning for $i=1$. Let $Z_i = \sqrt{n} Y_i$.
Proposition~\ref{proposition: a similar scaling of the
concentration inequalities} therefore implies that for an
arbitrary $\alpha
> 0$
\begin{eqnarray}
&& \pr( | \text{CF}_n(s) - \expectation [
\text{CF}_n(s)] | \geq \alpha ) \nonumber\\
&& = \pr( |Y_n - Y_0| \geq \alpha) \nonumber\\
&& = \pr( |Z_n - Z_0| \geq \alpha \sqrt{n}) \nonumber\\
&& \leq 2 \exp \left( -\frac{\alpha^2}{4} \, \Bigl(1 +
O\Bigl(\frac{1}{\sqrt{n}}\Bigr) \right) \label{eq:
OFDM-inequality1}
\end{eqnarray}
(since $\delta = \frac{\alpha}{2}$ and $\gamma = \frac{1}{2}$ in
the setting of Proposition~\ref{proposition: a similar scaling of
the concentration inequalities}). Note that the exponent of the
last concentration inequality is doubled as compared to the bound
that was obtained in \eqref{eq: Azuma's inequality for OFDM} via
Azuma's inequality, and the term which scales like
$O\Bigl(\frac{1}{\sqrt{n}}\Bigr)$ on the right-hand side of
\eqref{eq: OFDM-inequality1} is expressed explicitly for finite
$n$ (see \cite[Appendix~H]{Sason_submitted_paper}).

\subsection{Establishing Concentration via
McDiarmid's Inequality} In the following, McDiarmid's inequality
is applied to prove a concentration inequality for the crest
factor of OFDM signals. To this end, let us define
\begin{eqnarray*}
&& U \triangleq \max_{0 \leq t \leq T} \bigl|s(t; X_0, \ldots,
X_{i-1},
X_i, \ldots, X_{n-1})\bigr| \\
&& V \triangleq \max_{0 \leq t \leq T} \bigl|s(t; X_0, \ldots,
X'_{i-1}, X_i, \ldots, X_{n-1})\bigr|.
\end{eqnarray*}
Then, this implies that
\begin{eqnarray}
&& \hspace*{-0.5cm} |U-V| \leq \max_{0 \leq t \leq T} \bigl|s(t;
X_0, \ldots,
X_{i-1}, X_i, \ldots, X_{n-1}) \nonumber \\
&& \hspace*{2.5cm} - s(t; X_0, \ldots, X'_{i-1}, X_i, \ldots,
X_{n-1})\bigr| \nonumber \\[0.15cm]
&& \hspace*{0.8cm} = \max_{0 \leq t \leq T} \frac{1}{\sqrt{n}} \,
\Bigr|\bigl(X_{i-1} - X'_{i-1}\bigr) \exp\Bigl(\frac{j \, 2\pi i
t}{T}\Bigr)\Bigr| \nonumber \\[0.15cm]
&& \hspace*{0.8cm} = \frac{|X_{i-1} - X'_{i-1}|}{\sqrt{n}} \leq
\frac{2}{\sqrt{n}} \label{eq: bound2 on |U-V|}
\end{eqnarray}
where the last inequality holds since $|X_{i-1}|=|X'_{i-1}|=1$.
Hence, McDiarmid's inequality in Theorem~\ref{theorem: McDiarmid's
Inequality} implies that, for every $\alpha \geq 0$,
\begin{eqnarray}
\pr( | \text{CF}_n(s) - \expectation [ \text{CF}_n(s)] | \geq
\alpha ) \leq 2 \exp\Bigl(-\frac{\alpha^2}{2}\Bigr) \label{eq:
McDiarmid's inequality for OFDM}
\end{eqnarray}
which demonstrates concentration around the expected value. It is
noted that McDiarmid's inequality provides an improvement in the
exponent by a factor of~4 as compared to Azuma's inequality. It
also improves the exponent by a factor of~2 as compared to
Proposition~\ref{proposition: a similar scaling of the
concentration inequalities} in the considered case (where $\gamma
= \frac{1}{2}$).

The same kind of result applies easily to QAM-modulated OFDM
signals, since the RVs are bounded which therefore enables to get
a similar result to \eqref{eq: bound2 on |U-V|}.

\subsection{Establishing Concentration via Talagrand's
Inequality} In the following, Talagrand's inequality is applied to
prove a concentration inequality for the crest factor of OFDM
signals. Let us assume that $X_0, Y_0, \ldots, X_{n-1}, Y_{n-1}$
are i.i.d. bounded complex RVs, and also for simplicity
$$|X_i|=|Y_i|=1.$$ In order to apply Talagrand's inequality to
prove concentration, note that
\begin{eqnarray*}
&& \max_{0 \leq t \leq T} \bigl| \, s(t; X_0, \ldots,
X_{n-1})\bigr| -
\max_{0 \leq t \leq T} \bigl| \, s(t; Y_0, \ldots, Y_{n-1})\bigr| \\
&& \leq \max_{0 \leq t \leq T} \bigl| \, s(t; X_0, \ldots,
X_{n-1}) - s(t; Y_0, \ldots, Y_{n-1})\bigr| \\
&& \leq \frac{1}{\sqrt{n}} \; \left| \, \sum_{i=0}^{n-1} (X_i -
Y_i) \exp\Bigl(\frac{j \, 2\pi i
t}{T}\Bigr) \right| \\
&& \leq \frac{1}{\sqrt{n}} \; \sum_{i=0}^{n-1} |X_i - Y_i| \\
&& \leq \frac{2}{\sqrt{n}} \sum_{i=0}^{n-1} I_{\{x_i \neq y_i\}}
\\[0.1cm]
&& = 2 d_a(X,Y)
\end{eqnarray*}
where
\begin{equation}
a \triangleq \bigl(\frac{1}{\sqrt{n}}, \ldots,
\frac{1}{\sqrt{n}}\bigr) \label{eq: fixed vector a}
\end{equation}
is a non-negative unit-vector of length $n$ (note that $a$ in this
case is independent of $x$). Hence, Talagrand's inequality in
Theorem~\ref{theorem: Talagrand's Inequality} implies that, for
every $\alpha \geq 0$,
\begin{eqnarray}
\pr( | \text{CF}_n(s) - m_n | \geq \alpha ) \leq 4
\exp\Bigl(-\frac{\alpha^2}{16}\Bigr), \quad \forall \, \alpha
> 0 \label{eq: Talagrand's inequality for OFDM}
\end{eqnarray}
where $m_n$ is the median of the crest factor for OFDM signals
that are composed of $n$ sub-carriers. This inequality
demonstrates the concentration of this measure around its median.
As a simple consequence of \eqref{eq: Talagrand's inequality for
OFDM}, one obtains the following result.
\begin{corollary}
The median and expected value of the crest factor differ by at
most a constant, independently of the number of sub-carriers $n$.
\end{corollary}
\begin{proof}
By Talagrand's inequality in \eqref{eq: Talagrand's inequality for
OFDM}, it follows that
\begin{eqnarray*}
&& | \expectation [ \text{CF}_n(s)] - m_n | \\
&& \leq \expectation \left| \text{CF}_n(s) - m_n \right| \\
&& \stackrel{\text{(a)}}{=} \int_{0}^{\infty} \pr( |
\text{CF}_n(s) - m_n | \geq \alpha ) \;
d\alpha \\
&& \leq \int_{0}^{\infty} 4 \exp\Bigl(-\frac{\alpha^2}{16}\Bigr)
\; d\alpha \\
&& = 8 \sqrt{\pi}
\end{eqnarray*}
where equality~(a) holds since for a non-negative random variable
$Z$
\begin{equation*}
\expectation[Z] = \int_0^{\infty} \pr(Z \geq t) \, dt.
\end{equation*}
\end{proof}
\begin{remark}
This result applies in general to an arbitrary function~$f$
satisfying the condition in \eqref{eq: Talagrand's condition},
where Talagrand's inequality in \eqref{eq: Talagrand's inequality}
implies that (see, e.g., \cite[Lemma~4.6]{McDiarmid_tutorial})
$$ \, \big| \expectation [f(X)] - m \big| \leq 4 \sigma \sqrt{\pi}.$$
\end{remark}

\begin{remark}
By comparing \eqref{eq: Talagrand's inequality for OFDM} with
\eqref{eq: McDiarmid's inequality for OFDM}, it follows that
McDiarmid's inequality provides an improvement in the exponent.
This is consistent with Remark~\ref{remark: on a relation between
McDiardmid's inequality and a certain form of Talagrand's
inequality} and the fixed value of the non-negative normalized
vector in \eqref{eq: fixed vector a}.
\end{remark}

\section{Summary}
This paper derives four concentration inequalities for the
crest-factor (CF) of OFDM signals under the assumption that the
symbols are independent. The first two concentration inequalities
rely on Azuma's inequality and a refined version of it, and the
last two concentration inequalities are based on Talagrand's and
McDiarmid's inequalities. Although these concentration results are
weaker than some existing results from the literature (see
\cite{LitsynW06} and \cite{WunderB_IT}), they establish
concentration in a rather simple way and provide some insight to
the problem. The use of these bounding techniques, in the context
of concentration for OFDM signals, seems to be new. The improvement
of McDiarmid's inequality is by a factor of~4 in the exponent as
compared to Azuma's inequality, and by a factor of~2 as compared
to the refined version of Azuma's inequality in
Proposition~\ref{proposition: a similar scaling of the
concentration inequalities}. Note however
that Proposition~\ref{proposition: a similar scaling of the
concentration inequalities} may be in general tighter than
McDiarmid's inequality (if $\gamma < \frac{1}{4}$ in the setting
of Proposition~\ref{proposition: a similar scaling of the
concentration inequalities}). It also follows from Talagrand's
method that the median and expected value of the CF differ by at
most a constant, independently of the number of sub-carriers.

Some other new refined versions of Azuma's inequality were
introduced in \cite{Sason_submitted_paper}, followed by some
applications in information theory and communications. This work
is aimed to stimulate the use of some refined versions of
concentration inequalities, based on the martingale approach and
Talagrand's approach, in information-theoretic aspects.

The slides of the presentation of this work are available at
\cite{Sason_slides}.

\appendices

\section{Proof of the properties in \eqref{eq: properties for OFDM signals} for OFDM signals}
\label{appendix: OFDM} Consider an OFDM signal from
Section~\ref{subsection: background on OFDM}. The sequence in
\eqref{eq: martingale sequence for OFDM} is a martingale due to
basic properties of martingales. From \eqref{eq: CF}, for every $i
\in \{0, \ldots, n\}$
\begin{eqnarray*}
&& \hspace*{-0.7cm} Y_i = \expectation \Bigl[ \, \max_{0 \leq t
\leq T} \bigl|s(t; X_0, \ldots, X_{n-1}) \bigr| \Big| \, X_0,
\ldots, X_{i-1}\Bigr].
\end{eqnarray*}
The conditional expectation for the RV $Y_{i-1}$ refers to the
case where only $X_0, \ldots, X_{i-2}$ are revealed. Let
$X'_{i-1}$ and $X_{i-1}$ be independent copies, which are also
independent of $X_0, \ldots, X_{i-2}, X_i, \ldots, X_{n-1}$. Then,
for every $1 \leq i \leq n$,
\begin{eqnarray*} && \hspace*{-0.7cm} Y_{i-1} =
\expectation \Bigl[ \, \max_{0 \leq t \leq T} \bigl|s(t; X_0,
\ldots, X'_{i-1}, X_i, \ldots,
X_{n-1})\bigr| \\
&& \hspace*{4cm} \Big| \, X_0, \ldots, X_{i-2} \Bigr] \\
&& = \expectation \Bigl[ \, \max_{0 \leq t \leq T} \bigl|s(t; X_0,
\ldots, X'_{i-1}, X_i, \ldots,
X_{n-1})\bigr| \\
&& \hspace*{4cm} \Big| \, X_0, \ldots, X_{i-2}, X_{i-1} \Bigr].
\end{eqnarray*}
Since $|\expectation(Z)| \leq E(|Z|)$, then for $i \in \{1,
\ldots, n\}$

\vspace*{-0.2cm} \small
\begin{equation}
\hspace*{-0.2cm} |Y_i - Y_{i-1}| \leq \expectation_{X'_{i-1}, X_i,
\ldots, X_{n-1}} \Bigl[|U-V| \; \Big| \; X_0, \ldots, X_{i-1}
\Bigr] \label{eq: bounded differences for the sequence Y}
\end{equation}
\normalsize where
\begin{eqnarray*}
&& U \triangleq \max_{0 \leq t \leq T} \bigl|s(t; X_0, \ldots,
X_{i-1},
X_i, \ldots, X_{n-1})\bigr| \\
&& V \triangleq \max_{0 \leq t \leq T} \bigl|s(t; X_0, \ldots,
X'_{i-1}, X_i, \ldots, X_{n-1})\bigr|.
\end{eqnarray*}
From \eqref{eq: OFDM signal}
\begin{eqnarray}
&& \hspace*{-1.5cm} |U-V| \leq \max_{0 \leq t \leq T} \bigl|s(t;
X_0, \ldots,
X_{i-1}, X_i, \ldots, X_{n-1}) \nonumber \\
&& \hspace*{1cm} - s(t; X_0, \ldots, X'_{i-1}, X_i, \ldots,
X_{n-1})\bigr| \nonumber \\[0.15cm]
&& \hspace*{-0.2cm} = \max_{0 \leq t \leq T} \frac{1}{\sqrt{n}} \,
\Bigr|\bigl(X_{i-1} - X'_{i-1}\bigr) \exp\Bigl(\frac{j \, 2\pi i
t}{T}\Bigr)\Bigr| \nonumber \\[0.15cm]
&& \hspace*{-0.2cm} = \frac{|X_{i-1} - X'_{i-1}|}{\sqrt{n}}.
\label{eq: bound on |U-V|}
\end{eqnarray}
By assumption, $|X_{i-1}| = |X'_{i-1}| = 1$, and therefore a.s.
$$|X_{i-1} - X'_{i-1}| \leq 2 \Longrightarrow |Y_i - Y_{i-1}| \leq \frac{2}{\sqrt{n}}.$$
In the following, an upper bound on the conditional variance
$$\text{Var}(Y_i \, | \, \mathcal{F}_{i-1}) = \expectation \bigl[
(Y_i - Y_{i-1})^2 \, | \, \mathcal{F}_{i-1} \bigr]$$ is obtained.
Since $\bigl(\expectation(Z)\bigr)^2 \leq \expectation(Z^2)$ for a
real-valued RV $Z$, then from \eqref{eq: bounded differences for
the sequence Y} and \eqref{eq: bound on |U-V|}
\begin{equation*}
\vspace*{-0.1cm} \expectation\bigl[(Y_i - Y_{i-1})^2 \, |
\mathcal{F}_{i-1}\bigr] \leq \frac{1}{n} \cdot
\expectation_{X'_{i-1}} \bigl[|X_{i-1} - X'_{i-1}|^2 \, | \,
\mathcal{F}_i \bigr]
\end{equation*}
where $\mathcal{F}_i$ is the $\sigma$-algebra that is generated by
$X_0, \ldots, X_{i-1}$. Due to a symmetry argument of the PSK
constellation, then it follows that
\begin{eqnarray*}
&& \hspace*{-2.6cm} \expectation \bigl[ (Y_i - Y_{i-1})^2 \, |
\, \mathcal{F}_{i-1} \bigr] \\
&& \hspace*{-2.6cm} \leq \frac{1}{n} \, \expectation_{X'_{i-1}}
\bigl[|X_{i-1} - X'_{i-1}|^2 \, | \, \mathcal{F}_i \bigr] \\
&& \hspace*{-2.6cm} = \frac{1}{n} \, \expectation \bigl[
|X_{i-1} - X'_{i-1}|^2 \, | \, X_0, \ldots, X_{i-1} \bigr] \\
&& \hspace*{-2.6cm} = \frac{1}{n} \, \expectation \bigl[
|X_{i-1} - X'_{i-1}|^2 \, | \, X_{i-1} \bigr] \\
&& \hspace*{-2.6cm} = \frac{1}{n} \, \expectation \Bigl[ |X_{i-1}
- X'_{i-1}|^2 \, | \, X_{i-1} = e^{\frac{j \pi}{M}} \Bigr] \\
&& \hspace*{-2.6cm} = \frac{1}{nM} \sum_{l=0}^{M-1}
\Big| \, e^{\frac{j \pi}{M}}-e^{\frac{j (2l+1)\pi}{M}}\Big|^2 \\
&& \hspace*{-2.6cm} = \frac{4}{nM} \sum_{l=1}^{M-1} \sin^2 \Bigl(
\frac{\pi l}{M} \Bigr) = \frac{2}{n}.
\end{eqnarray*}


\begin{thebibliography}{99}
\bibitem{Azuma}
K. Azuma, ``Weighted sums of certain dependent random variables,''
{\em Tohoku Mathematical Journal}, vol.~19, pp.~357--367, 1967.
\bibitem{Dembo_Zeitouni}
A. Dembo and O. Zeitouni, {\em Large Deviations Techniques and
Applications}, Springer, second edition, 1997.
\bibitem{Hoeffding}
W. Hoeffding, ``Probability inequalities for sums of bounded
random variables,'' {\em Journal of the American Statistical
Association}, vol.~58, no.~301, pp.~13--30, March 1963.
\bibitem{Korada_Macris_IT10}
S. B. Korada and N. Macris, ``Tight bounds on the capacity of
binary input random CDMA systems,'' {\em IEEE Trans. on
Information Theory}, vol.~56, no.~11, pp.~5590--5613,
November~2010.
\bibitem{LitsynW06}
S. Litsyn and G. Wunder, ``Generalized bounds on the crest-factor
distribution of OFDM signals with applications to code design,''
{\em IEEE Trans. on Information Theory}, vol.~52, pp.~992--1006,
March 2006.
\bibitem{McDiarmid_tutorial}
C. McDiarmid, ``Concentration,'' {\em Probabilistic Methods for
Algorithmic Discrete Mathematics}, pp.~195--248, Spinger, 1998.
\bibitem{Molisch_book}
A. F. Molisch, {\em Wireless Communications}, John Wiley and Sons,
2005.
\bibitem{RiU_book}
T. Richardson and R. Urbanke, {\em Modern Coding Theory},
Cambridge University Press, 2008.
\bibitem{SalemZ}
R. Salem and A. Zygmund, ``Some properties of trigonometric series
whose terms have random signs,'' {\em Acta Mathematica}, vol.~91,
no.~1, pp.~245--301, 1954.
\bibitem{Sason_submitted_paper}
I. Sason, ``On refined versions of the Azuma-Hoeffding inequality
with applications in information theory,'' last updated in July 2012. 
[Online]. Available: \url{http://arxiv.org/pdf/1111.1977v5.pdf}.
\bibitem{Sason_slides}
I. Sason, {\em slides of the presentation of this work}, presented 
at the {\em 2011 IEEE International Symposium on Wireless
Communication Systems (ISWCS~2011)}, Aachen, Germany, November 2011.
[Online]. Available: \url{http://webee.technion.ac.il/people/sason/ISWCS11_presentation.pdf}.
\bibitem{Talagrand95}
M. Talagrand, ``Concentration of measure and isoperimteric
inequalities in product spaces,'' {\em Publications
Math\'{e}matiques de l'I.H.E.S}, vol.~81, pp.~73--205, 1995.
\bibitem{Talagrand96}
M. Talagrand, ``A new look at independence,'' {\em Annals of
Probability}, vol.~24, no.~1, pp.~1--34, January 1996.
\bibitem{WunderB_IT}
G. Wunder and H. Boche, ``New results on the statistical distribution
of the crest-factor of OFDM signals,'' {\em IEEE Trans. on Information
Theory}, vol.~49, no.~2, pp.~488--494, February 2003.
\bibitem{Volterra_IT_March11}
K. Xenoulis and N. Kalouptsidis, ``Achievable rates for nonlinear
Volterra channels,'' {\em IEEE Trans. on Information Theory}, vol.~57,
no.~3, pp.~1237--1248, March 2011.
\end{thebibliography}
\end{document}